\documentclass[sigplan,screen]{acmart}
\usepackage{multirow}
\usepackage{mathtools}
\usepackage{amsthm}
\usepackage{amsmath}
\usepackage{subcaption}

\usepackage{tabularx}
\usepackage{array}
\usepackage{csvsimple}
\usepackage{adjustbox}

\usepackage{algorithm}
\usepackage{algpseudocode}

\usepackage{subcaption}

\setcopyright{none}
\renewcommand\footnotetextcopyrightpermission[1]{}

\newcommand{\red}[1]{\textcolor{black}{#1}}
\newcommand{\csvtext}[1]{\texttt{\detokenize\expandafter{#1}}}

\title{
 \red{Compiling Together: High-Throughput Distributed Quantum Computing via Multi-Compilation}
}
\author{\red{Yipei Liu$^{1}$, Sen Zhang$^{1}$, Zebo Yang$^{2}$, Lei Yang$^{1}$}}
\affiliation{\institution{\red{$^{1}$George Mason University \quad $^{2}$Florida Atlantic University}}\country{}}

\newtheorem{definition}{Definition}

\newtheorem{theorem}{Theorem}

\begin{document}

\begin{abstract}
Quantum computing is a promising paradigm for problems that are challenging for classical machines, but realizing that promise requires far more qubits than a single processor can offer. Distributed quantum computing (DQC) scales out by connecting multiple quantum processing units (QPUs), at the cost of making entanglement the scarce resource: every remote gate consumes a Bell pair, and inter-QPU links generate Bell pairs at finite rates orders of magnitude slower than local gates. Since quantum programs are executed repeatedly, this rate bounds how fast shots complete and thus how fast results are obtained. 
Existing DQC compilers emit a single implementation per circuit, so throughput is capped by its busiest link while other links stay idle. We observe that alternative compilations of the same circuit are logically equivalent yet stress different links; executing them concurrently and pooling their samples converts idle Bell pairs into additional shots. 
We formulate the joint selection of compilations and allocation of shots under per-link Bell-pair capacities as Candidate-Constrained Max-Shot Allocation (CMA), prove it NP-hard, and solve it with a dynamic program and its approximate variant (\red{AppDP}), a compact MILP, and a greedy heuristic (\red{Effi}).
In simulation on six-QPU networks, multi-compilation raises throughput by 2–4.5× over the single compilation and correspondingly improves output fidelity under equal Bell-pair budgets. On hardware, it increases measured fidelity by up to 93\% and reaches the single-compilation fidelity in up to 10× less time. 
Scaling to 36 QPUs and 144-qubit circuits, MILP improves throughput over the single compilation by 76.8\% on average across 48 configurations while \red{Effi} allocates in milliseconds.
\end{abstract}

\maketitle

\fancyhead{} 

\section{Introduction}
\label{sec:intro}

Quantum computing has emerged as a promising computing paradigm with the potential to tackle computational problems that are challenging for classical machines~\cite{preskill2018nisq}. 
However, realizing this potential at practically relevant scales requires quantum systems with substantially more qubits, while scaling a single monolithic quantum processor remains increasingly difficult~\cite{caleffi2024survey,barral2025review}.
Distributed quantum computing (DQC) offers a scalable alternative by connecting multiple quantum processing units (QPUs) through communication links and executing a quantum circuit across them~\cite{cirac1999distributed,caleffi2024survey,barral2025review}. 
Executing gates across different QPUs, however, requires inter-QPU communication supported by shared entanglement resources. Unlike local quantum gates, these entanglement resources must be generated over communication links at finite rates, making communication capacity a critical constraint on distributed execution.


Distribution, however, comes with the cost of inter-QPU entanglement. Remote operations between qubits located on different QPUs consume Bell pairs along their communication paths, while these Bell pairs can only be generated at finite rates. Reported entanglement-generation rates include approximately $9.7$ and $182$ pairs per second~\cite{stephenson2020highrate,main2025distributed}, corresponding to average generation intervals of approximately $103$ and $5.5$\,ms per Bell pair, respectively, whereas local superconducting gates typically complete in tens of nanoseconds~\cite{acharya2025below,abughanem2025ibm}. 
This disparity can make inter-QPU communication a throughput bottleneck in distributed execution. 
Prior work also identifies entanglement generation as a critical constraint in fault-tolerant DQC~\cite{jacinto2026network,stack2026transversal}. In communication-dominated DQC systems, the Bell-pair generation rate can therefore determine the achievable execution throughput.

The throughput limitation is particularly consequential because quantum programs must be executed repeatedly to obtain reliable results. A single circuit execution produces only one measurement sample, or \emph{shot}, from the output distribution. As quantum circuits scale to more qubits, the measurement space grows exponentially, with up to $2^n$ possible bitstrings for an $n$-qubit circuit, increasing the sampling burden required to characterize the output distribution or estimate quantities of interest. Many quantum algorithms, including VQE and QAOA, therefore rely on repeated circuit executions to estimate expectation values and other statistical quantities~\cite{peruzzo2014vqe,farhi2014qaoa,preskill2018nisq}. Consequently, when inter-QPU communication limits execution throughput, the rate at which Bell-pair resources are converted into completed shots directly affects how quickly the quantum computational results can be obtained.

Achieving high shot throughput therefore depends not only on how much entanglement the network can generate, but also on how effectively a compiled circuit can utilize that capacity. Existing DQC compilation flows typically produce a single execution plan by partitioning the circuit, mapping its components to QPUs, and routing remote operations according to an optimization objective such as communication cost~\cite{andresmartinez2019automated,baker2020timesliced,ferrari2021compiler,wu2022autocomm}. However, Bell-pair capacity is distributed across individual communication links, and a single compilation may place most of its communication demand on only a subset of them. Since multiple communication links can generate Bell pairs in parallel, a single compilation may leave substantial entanglement-generation capacity unused if its communication demand is concentrated on only a subset of the links. Different compilations of the same circuit can exhibit different communication patterns and therefore stress different links. Executing multiple such compilations together can exploit complementary link capacities, increasing the aggregate number of completed shots per unit time beyond what any single compilation can achieve.


Realizing this potential requires jointly selecting compiled candidates and allocating shots according to their per-link Bell-pair demands. 
Prior work has explored multiple mappings to mitigate device errors~\cite{tannu2019ensemble} and scheduled circuits across QPUs to reduce makespan~\cite{bahrani2024resource}, but has not addressed how to combine alternative compilations to maximize shot throughput under per-link Bell-pair capacity constraints. Equivalent compilations produce the same ideal output distribution, allowing their quantum measurement samples to be pooled to reduce sampling error.

Two challenges arise in exploiting this opportunity. First, candidate selection and shot allocation are coupled: Candidates cannot be selected based only on their individual throughput. Two high-throughput candidates may rely on the same communication links and therefore compete for the same Bell-pair capacity, while candidates using different links can complement one another and achieve a higher combined throughput. The allocation must therefore jointly determine which candidates to execute and how many shots to assign to each under the per-link capacity constraints. Second, different compilations may incur different communication costs and noise, so increasing the number of completed shots does not necessarily improve the accuracy of the pooled result. The benefits must therefore be evaluated under noisy distributed execution.

To address these challenges, we jointly select compiled candidates and allocate their shots to exploit complementary link capacities. We evaluate both the resulting throughput gains and their impact on output fidelity under noisy distributed execution. Our main contributions are as follows:

\begin{itemize}
  \item We introduce multi-compilation execution for DQC, exploiting
  complementary Bell-pair generation resource across alternative compilations to increase execution shot throughput.

  \item We formulate Candidate-Constrained Max-Shot Allocation (CMA),   prove its NP-hardness, and design multiple allocation strategies that balance solution quality
and optimization overhead.

  \item We validate multi-compilation execution through distributed
  simulator and IBM quantum hardware. Across multiple circuit workloads and network topologies, it improves execution shot throughput over single compilation and fidelity under equal Bell-pair budgets.

  \item We demonstrate scalability up to 36 QPUs and 144-qubit circuits.
  MILP achieves the highest throughput across the evaluated configurations,
  \red{AppDP} retains most of the MILP throughput with approximation, while
  \red{Effi} provides millisecond-scale allocation overhead.
\end{itemize}

\section{Background and Related Work}
\label{sec:background}

\subsection{Distributed Quantum Computing}
\label{sec:background_dqc}

Distributed quantum computing (DQC) connects multiple quantum processing units
(QPUs) through quantum and classical communication links, allowing a logical
circuit to be partitioned and executed across multiple processors
~\cite{cirac1999distributed,caleffi2024survey,barral2025review}.
Rather than scaling a single monolithic processor, DQC provides a modular
approach in which several smaller QPUs cooperate to execute circuits whose
qubit or connectivity requirements exceed the capability of an individual
device.

\begin{figure}[b]
\centering
\includegraphics[width=\columnwidth]{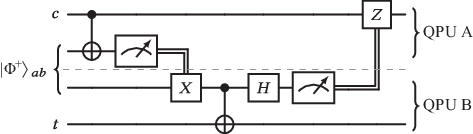}
\Description{Quantum circuit of a remote CNOT. The control qubit c and the
communication qubit a sit on QPU A; the communication qubit b and the target t sit
on QPU B; a and b start in a Bell pair. A CNOT from c to a, a measurement of a, and
a conditional X on b form the cat-entangler. A CNOT from b to t applies the gate.
A Hadamard and measurement on b with a conditional Z on c form the
cat-disentangler.}
\caption{Remote 2-qubit gate (telegate) between data qubits on two QPUs using a Bell pair and classical communication.}
\label{fig:bg-telegate}
\end{figure}

To execute a circuit in a DQC system, logical qubits are first assigned to
different QPUs. Gates whose operands reside on the same QPU are executed
locally, while a two-qubit gate whose operands are placed on different QPUs
requires an inter-QPU operation. Such remote operations can be implemented
using either \emph{telegate} or \emph{teledata}. Telegate directly realizes a
remote two-qubit gate through shared entanglement, local operations,
measurements, and classical feed-forward
~\cite{eisert2000optimal,gottesman1999teleportation,
yimsiriwattana2004generalized}. 
Figure~\ref{fig:bg-telegate} illustrates a remote 2-qubit gate implemented through telegate. The control and target data qubits remain on their original QPUs, while a Bell pair provides the nonlocal quantum resource. Local gates, intermediate measurements, and classical feed-forward corrections first correlate the remote communication qubit with the control and then disentangle the auxiliary qubits after the target operation is applied. Thus, the nonlocal gate is realized without physically moving either data qubit, at the cost of consuming one Bell pair and exchanging classical measurement outcomes between the two QPUs. Teledata follows a different strategy: it teleports a data qubit to another QPU, where subsequent gates can then be executed locally~\cite{bennett1993teleporting}. Both approaches therefore convert nonlocal circuit operations into local quantum operations supported by Bell-pair resources and classical communication.

Both mechanisms rely on entanglement shared between communicating QPUs. Consequently, the placement of logical qubits and the routing of remote operations determine which communication links are used and how many Bell pairs must be supplied by each link. Different compilations of the same logical circuit can therefore induce different communication demands across the DQC network. Existing DQC compilers optimize placement and routing to reduce the number or cost of remote operations~\cite{andresmartinez2019automated,baker2020timesliced,ferrari2021compiler,cuomo2023optimized,wu2022autocomm,wu2023qucomm}, but each produces a single execution plan per circuit. Related efforts execute several variants of a circuit to diversify device errors~\cite{tannu2019ensemble} and co-schedule multiple programs on a single device~\cite{das2019multiprogramming} or across DQC interconnects~\cite{bahrani2024resource}; none of them combines equivalent compilations to exploit per-link Bell-pair capacity.

\subsection{Bell-Pair Communication}
\label{sec:background_bell}

Remote operations require Bell pairs established between communication qubits
at different QPUs. A Bell pair,
$|\Phi^+\rangle=(|00\rangle+|11\rangle)/\sqrt{2}$, provides the shared
entanglement consumed by protocols such as telegate and teledata. In many
quantum-network implementations, entanglement generation is probabilistic:
communication qubits repeatedly attempt to establish a Bell pair through a
physical link, and successful generation is heralded by a measurement event.

As a result, Bell pairs are not available instantaneously. Instead, each
communication link supplies them at a finite generation rate. Experimental
quantum links have demonstrated generation rates ranging from tens to hundreds
of Bell pairs per second depending on the hardware and operating conditions
~\cite{stephenson2020highrate,main2025distributed}. We therefore characterize
each communication link $e$ by a Bell-pair generation rate $R_e$, measured in
pairs/s.

Importantly, Bell-pair generation is associated with individual communication
links rather than a single network-wide resource pool. Different links can
generate and consume Bell pairs concurrently, while unused capacity on one
link cannot directly compensate for insufficient capacity on another. 
This per-link communication model, consistent with models used in quantum-network routing~\cite{pant2019routing,shi2020concurrent}, forms the basis for the DQC resource formulation used in the following sections.

\section{{Observation and Motivation}}
\label{sec:motivation}
This section builds the case for multi-candidate execution in three steps.
Entanglement is the scarce resource of a DQC system, and a single compilation leaves
much of it idle (Section~\ref{sec:motiv_scarce}); the idle capacity is not abstract,
because under a fixed Bell-pair budget it is lost output fidelity
(Section~\ref{sec:motiv_accuracy}); and combining several compiled candidates
recovers that fidelity, most strongly when entanglement is scarce
(Section~\ref{sec:motiv_multi}).

\begin{figure}[t]
    \centering
    \includegraphics[width=\columnwidth]{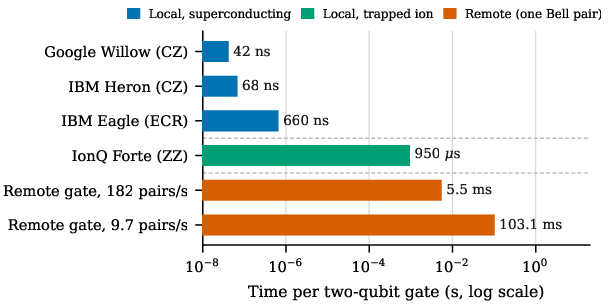}
    \caption{Comparison of local two-qubit gate latency and Bell-pair generation time for remote quantum operations.}
    \label{fig:motiv-cost}
\end{figure}

\begin{figure*}[h]
    \centering
    \includegraphics[width=\textwidth]{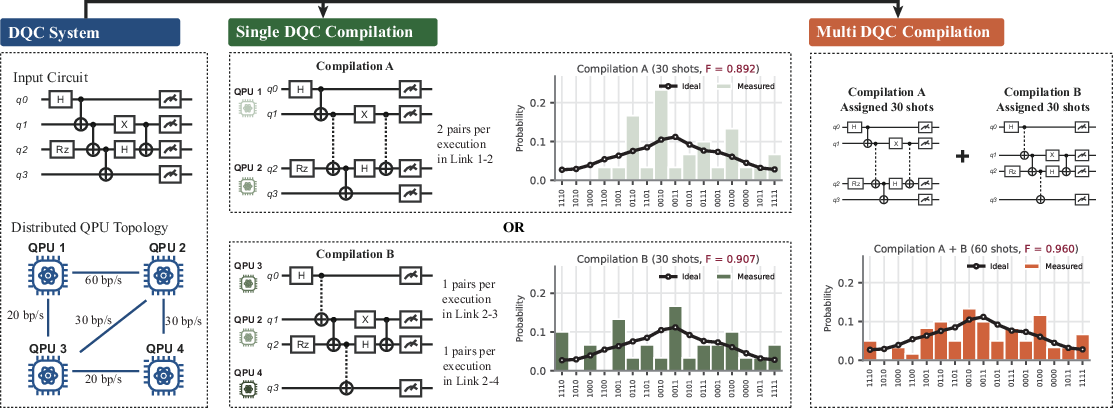}
    \caption{A four-QPU example illustrating complementary compilations and the fidelity improvement obtained by merging their measurement results. Under the same Bell-pair capacities, a single compilation exploits only part of the DQC resources and executes 30 shots, whereas multiple compilations achieve 60 shots with higher fidelity.}
        \label{fig:example1}
\end{figure*}

\subsection{Bell-Pair Generation Limits Throughput}
\label{sec:motiv_scarce}

Bell-pair generation for inter-QPU communication can be substantially slower than local two-qubit operations. For example, a CZ gate takes about $42\,\mathrm{ns}$ on Google's Willow processor~\cite{acharya2025below} and $68\,\mathrm{ns}$ on IBM Heron devices~\cite{abughanem2025ibm}, while two-qubit gates on trapped-ion systems can operate on a longer timescale, such
as about $950\,\mu\mathrm{s}$ on IonQ Forte~\cite{chen2024forte}. By comparison, the distributed trapped-ion system ~\cite{main2025distributed} generated remote entanglement at an average rate of $9.7\,\mathrm{s}^{-1}$, corresponding to about $103\,\mathrm{ms}$ per Bell pair. An earlier experiment on the same apparatus reported a rate of $182\,\mathrm{s}^{-1}$, or approximately $5.5\,\mathrm{ms}$ per Bell pair~\cite{stephenson2020highrate}. Figure~\ref{fig:motiv-cost}
compares these local gate times with the Bell-pair generation times required for inter-QPU communication. Although the measurements come from different hardware, they illustrate that remote entanglement generation is substantially slower than local quantum operations, motivating the communication-dominated regime considered in this work.

Under this regime, Bell-pair generation rates directly constrain circuit throughput. Each compiled candidate fixes the placement and routing of remote gates, and therefore induces a specific Bell-pair demand across the DQC network. Its execution rate can be limited by the most heavily used link even when substantial communication capacity remains available on other links. Since different DQC links can operate concurrently, multiple compiled candidates with complementary communication demands can utilize otherwise-idle links
simultaneously. Combining such candidates therefore provides an opportunity to better exploit the available entanglement resources and increase aggregate circuit throughput.



\subsection{Wasted Bell-Pair Capacity Leads to Fidelity Loss}
\label{sec:motiv_accuracy}
Figure~\ref{fig:example1} gives a concrete four-qubit example, whose measurement space contains $2^4=16$ possible bitstrings. In Compilation~A, $q_0$ and $q_1$
are placed on QPU~1, while $q_2$ and $q_3$ are placed on QPU~2. Each execution
therefore consumes two Bell pairs on link $1$--$2$. With a generation rate of 60 Bell pairs/s, this compilation A can execute 30 shots/s. Compilation~B uses a
different partition: $q_0$ is placed on QPU~3, $q_1$ and $q_2$ on QPU~2, and
$q_3$ on QPU~4. Each execution consumes one Bell pair on link $2$--$3$ and one on link $2$--$4$, both of which provide 30 Bell pairs/s, again allowing 30 shots/s. Importantly, the two compilations use different communication links, so their Bell-pair demands do not compete under the independent-link model. Executing them together therefore increases the aggregate execution rate from
30 to 60 shots/s.

The corresponding measurement distributions further show the benefit of this additional sampling. For Compilation~A, the measured probabilities of the 16 bitstrings, shown by the light-green bars, deviate noticeably from the ideal
distribution, resulting in a fidelity of 0.892. Compilation~B produces a different measured distribution, shown by the dark-green bars, with a fidelity of 0.907. Because the two compilations implement the same logical circuit, their
measurement samples represent the same set of quantum output bitstrings and can be merged directly. Combining the 30 shots from A and the 30 shots from B produces 60 samples in total; the resulting orange distribution more closely follows the
ideal black curve and improves the fidelity to 0.960. Thus, complementary compilations can convert otherwise-idle communication capacity into additional measurement samples and a more accurate estimate of the quantum output
distribution.


\subsection{Multi-compilation Outperforms Single}
\label{sec:motiv_multi}

Figure~\ref{fig:motiv-resource} evaluates this opportunity in a more general setting. We compile an 8-qubit circuit into six compilations with different placements on the DQC network and allow them to execute concurrently. We vary the Bell-pair capacity of the whole DQC.
Under the same DQC resource, multi-compilation can utilize multiple communication links concurrently, completing more circuit executions than a single compilation. As shown in Figure~\ref{fig:motiv-resource}, these additional executions consistently translate into higher output fidelity. The improvement is most pronounced when Bell-pair resources are scarce and gradually diminishes as the available capacity becomes sufficient for a single compilation to complete the requested shots.

\begin{figure}[t]
    \centering
    \includegraphics[width=\columnwidth]{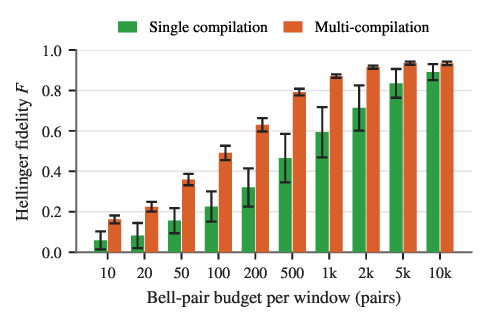}
    \caption {Fidelity under different Bell-pair budgets for single-
and multi-compilation execution. Bars show mean $\pm$ standard deviation
over 20 independent replicates.}
    \label{fig:motiv-resource}
\end{figure}

\section{DQC System Model}
\label{sec:system}

\subsection{Network Resources and Compiled Candidates}

We consider a DQC platform represented by a network $G=(V,E)$,
where vertices are QPUs and edges are inter-QPU communication links.
Each link $e$ supplies usable Bell pairs at a finite rate
$R_e\in\mathbb{Z}_{\geq0}$, measured in pairs/s. Link rates may be
heterogeneous. We focus on a communication-limited setting in which local
computation and control resources do not impose additional binding constraints.
The model accounts for aggregate communication demand; it does not specify
the timing of individual gates or a physical execution schedule.

For a logical circuit $\mathcal{Q}$, compilation produces a fixed candidate
set $C=\{1,\ldots,m\}$. Each candidate implements $\mathcal{Q}$ through a
particular partition, QPU mapping, and routing of remote operations.
Candidate $i$ consumes $a_{i,e}\in\mathbb{Z}_{\geq0}$ Bell pairs on link $e$
per shot, giving the resource vector $\mathbf{a}_i=(a_{i,e})_{e\in E}$.
We assume $\mathbf{a}_i\neq\mathbf{0}$ for every $i$: each distributed
candidate uses at least one communication link. This assumption also ensures
that finite link rates bound every candidate's allocation.
Candidate generation is outside the allocation problem; the optimizer takes
the resulting resource vectors as input.

\subsection{Allocation and Resource Complementarity}
\label{sec:shot_allocation}

We allocate throughput in integer units of one shot/s.
Thus, $x_i\in\mathbb{Z}_{\geq0}$ is the rate assigned to candidate $i$,
the aggregate throughput is $S(x)=\sum_{i\in C}x_i$, and the demand on
link $e$ is $U_e(x)=\sum_{i\in C}a_{i,e}x_i$ pairs/s.
Integrality is an allocation granularity adopted in this work, rather than
a requirement on arbitrary long-run average execution rates.
All allocations below refer to this discretized rate model.


\section{Problem Formulation}
\label{sec:problem}

\subsection{Problem Statement}
\label{sec:problem_definition}

Given a QPU network $G=(V,E)$, the Bell-pair generation rate $R_e$ of each
communication edge $e\in E$, and a
set of compiled candidates $C$, our objective is to determine the execution
rate $x_i$ assigned to each candidate $i\in C$ so as to maximize the aggregate
execution throughput of the logical quantum circuit.

Each compiled candidate may exhibit a different communication resource
consumption pattern due to its circuit partitioning, QPU mapping, and routing
decisions. Consequently, assigning executions to different candidates induces
different Bell-pair demands across the communication links. The allocation
must therefore ensure that the aggregate Bell-pair consumption rate on every communication edge does not exceed its available Bell-pair generation rate.
At the same time, different candidates may compete for Bell-pair generation capacity on the same DQC link or exhibit complementary communication demands across different links. We therefore jointly determine which compiled candidates to execute and how many shots to assign to each selected candidate in order to maximize the aggregate circuit-execution throughput. We formulate this joint candidate-selection and shot-allocation task as the Candidate-Constrained Max-Shot Allocation (CMA) problem.

\subsection{Allocation Problem}
\label{sec:feasibility}

The following definitions specify the communication resource and candidate selection constraints of the allocation problem.

\begin{definition}[Edge-resource feasibility]
\label{def:edge_feasibility}

For an allocation $x$, the aggregate Bell-pair consumption rate on edge $e$ is
\begin{equation}
U_e(x)=\sum_{i\in C}a_{i,e}x_i.
\label{eq:edge_usage}
\end{equation}
The allocation is edge-resource feasible if
$U_e(x)\leq R_e$ for every $e\in E$ in DQC system.
\end{definition}

This condition ensures that the aggregate Bell-pair demand of all executed candidates does not exceed the generation rate of any DQC link. The constraint is imposed separately on each edge because unused capacity on one link cannot compensate for insufficient capacity on another.

\begin{definition}[Selected compiled candidates]
\label{def:selected_candidates}

For an allocation $x$, let
$C_{\mathrm{select}}(x)=\{i\in C\mid x_i>0\}$ denote the set of active
compiled candidates, and let
$N_{\mathrm{select}}(x)=|C_{\mathrm{select}}(x)|$.
Given a candidate limit $K$, a valid allocation must satisfy
$N_{\mathrm{select}}(x)\leq K$.
\end{definition}

The candidate limit controls how many compiled implementations of the same
logical circuit may jointly participate in execution. This limit captures
the practical overhead of maintaining and scheduling multiple compiled
implementations.

\begin{definition}[Candidate-Constrained Max-Shot Allocation]
Given a DQC network $G=(V,E)$, Bell-pair generation rates
$\{R_e\}_{e\in E}$, a compiled candidate set $C$, per-shot
resource requirements $\{a_{i,e}\}$, and a candidate limit $K$,
the Candidate-Constrained Max-Shot Allocation (CMA) problem determines
an integer execution-rate allocation that maximizes the aggregate
circuit-execution throughput:
\begin{align}
\max_x\quad
& S(x)=\sum_{i\in C}x_i
\label{eq:CMA_obj}\\
\mathrm{s.t.}\quad
& \sum_{i\in C}a_{i,e}x_i\leq R_e,
&& \forall e\in E,
\label{eq:CMA_edge}\\
& \left|\{i\in C:x_i>0\}\right|\leq K,
\label{eq:CMA_candidate}\\
& x_i\in\mathbb{Z}_{\geq0},
&& \forall i\in C.
\label{eq:CMA_domain}
\end{align}
\end{definition}

Constraint~\eqref{eq:CMA_edge} ensures edge-resource feasibility
as defined in Definition~\ref{def:edge_feasibility}, while
Constraint~\eqref{eq:CMA_candidate} enforces the candidate limit
in Definition~\ref{def:selected_candidates}.
The optimal objective value is the maximum feasible
circuit-execution throughput under these constraints.

\subsection{Computational Hardness}
\label{sec:hardness}

For a single compiled candidate, the maximum executable shot rate is determined directly by its bottleneck communication link. However, when multiple candidates participate, their resource demands interact across different links, and the optimal allocation cannot generally be determined from their individual throughput values alone. We next establish the computational hardness of CMA.

\begin{definition}[Decision version]
Given a QPU network $G=(V,E)$, Bell-pair generation rates
$\{R_e\}_{e\in E}$, a compiled candidate set $C$, a candidate limit $K$, and a
target throughput $T$, the decision version asks whether there exists an
integer allocation $x$ satisfying Constraints~\eqref{eq:CMA_edge}--
\eqref{eq:CMA_domain} such that $S(x)\geq T$.
\end{definition}

\begin{theorem}
The decision version is NP-complete. Consequently, the optimization version is
NP-hard.
\end{theorem}

\begin{proof}
The decision version belongs to NP because, given an allocation $x$, one can check all edge-capacity constraints, count the selected candidates, and compute $S(x)$ in polynomial time.

We reduce from the decision version of maximum set packing. Given a universe $U$, a collection of subsets $\mathcal{S}=\{S_1,\ldots,S_m\}$, and an integer target $T$, the set-packing decision problem asks whether there exist at least $T$ pairwise-disjoint subsets. From this instance, construct an allocation instance as follows. Create one compiled candidate $i$ for each subset $S_i$. For every element $u\in U$, create one communication edge $e_u$ with Bell-pair generation rate $R_{e_u}=1$. Candidate $i$ consumes one Bell pair on edge $e_u$ if $u\in S_i$, and consumes zero Bell pairs on that edge otherwise.
Additionally, create one private communication edge $p_i$ for each candidate
$i$, also with generation rate $R_{p_i}=1$, and let only candidate $i$ consume
one Bell pair on $p_i$. Set the candidate limit to $K=T$ and set the target
throughput to $T$.

The private edge $p_i$ forces $x_i\leq 1$ for every candidate. Therefore, any
feasible allocation with $S(x)\geq T$ must select exactly $T$ candidates, each
with one shot. If two selected candidates correspond to subsets that share an
element $u$, then both consume the unit-capacity edge $e_u$, violating
$R_{e_u}=1$. Hence the selected subsets must be pairwise disjoint. Conversely,
if the set-packing instance contains $T$ pairwise-disjoint subsets, assigning
one shot to the corresponding $T$ candidates satisfies every edge-capacity
constraint and achieves throughput $T$.

Thus, the constructed allocation instance is feasible with throughput at least
$T$ if and only if the original set-packing instance has a packing of size at
least $T$. The reduction is polynomial, so the decision version is
NP-complete. Therefore, the optimization version is NP-hard.
\end{proof}

\section{Optimization Methods}
\label{sec:methods}

We first develop a dynamic program that directly captures the candidate-selection and resource-allocation structure of CMA. 
We then present a compact MILP formulation that avoids explicit resource-state enumeration. Finally, we introduce \texttt{\red{Effi}}, a lightweight allocation approach designed to reduce the computational overhead of optimization.

\subsection{Dynamic Programming}
\label{sec:method_dp}

{Let $E^+=\{e\in E:\exists i\in C,\ a_{i,e}>0\}$ denote the set
of DQC links used by at least one candidate, and let $d=|E^+|$. Links outside $E^+$ do not constrain the allocation and are omitted from the resource state. For candidate $i$, define
$\mathbf{w}_i=(a_{i,e})_{e\in E^+}$ and let $\mathbf{B}=(R_e)_{e\in E^+}$. Since $R_e$, $a_{i,e}$, and $x_i$ are integer-valued, the resulting resource states are discrete. All vector inequalities below are componentwise. We assume that every candidate consumes at least one Bell pair on some communication link.}

{Let $F(j,k,\mathbf{b})$ denote the maximum number of shots achievable
using the first $j$ candidates, selecting at most $k$ of them, under
the available resource vector $\mathbf{b}$.
For candidate $j$, the maximum number of shots that can be assigned
under $\mathbf{b}$ is
$M_j(\mathbf{b})=
\min_{e\in E^+:w_{j,e}>0}
\lfloor b_e/w_{j,e}\rfloor$.
For $j\geq1$ and $k\geq1$,
\begin{equation}
\begin{aligned}
&F(j,k,\mathbf{b})=\max\Bigl\{F(j-1,k,\mathbf{b}),\\[-2pt]
&\hspace{2.5em}\max_{1\le z\le M_j(\mathbf{b})}\bigl[z+F(j-1,k-1,\mathbf{b}-z\mathbf{w}_j)\bigr]\Bigr\}.
\end{aligned}
\label{eq:dp_recurrence}
\end{equation}
The first term skips candidate $j$, while the second assigns it
$z$ positive shots, consumes $z\mathbf{w}_j$ resources, and uses one
candidate slot. The inner maximization is omitted when
$M_j(\mathbf{b})=0$. The boundary conditions are
$F(0,k,\mathbf{b})=0$ and $F(j,0,\mathbf{b})=0$, and the optimal
objective value is $F(m,K,\mathbf{B})$, where $m=|C|$ is the total number of compiled candidates and $K$ is the maximum number of candidates allowed in the allocation.}




\paragraph{Computational complexity.}
Let $M_{\max}=\max_{i\in C}M_i(\mathbf{B})$ and
$N_B=\prod_{e\in E^+}(B_e+1)$ denote the maximum feasible
shot assignment to any single candidate and the number of
possible resource vectors, respectively.
A direct tabulation contains $O(mKN_B)$ states.
Each state considers skipping the current candidate or
assigning it up to $M_{\max}$ positive shots, with
$O(d)$ resource-vector operations per transition.
The resulting time and space bounds are
\begin{equation}
\begin{aligned}
T_{\mathrm{DP}}
&=O\!\left(mK\,d\,(M_{\max}+1)N_B\right),\\
S_{\mathrm{DP}}
&=O\!\left(KN_B\right).
\end{aligned}
\label{eq:dp_complexity}
\end{equation}
The space bound uses rolling storage over the candidate
index and excludes storage for allocation reconstruction.
When all effective link capacities equal $B$,
$N_B=(B+1)^d$, so the resource grid grows exponentially
with the number of effective links.
For fixed $d$, the DP is pseudo-polynomial because its
complexity depends on the numerical capacity values
rather than only on their encoding lengths.
The sparse implementation stores only reachable states,
but the frontier can still grow substantially.

\paragraph{Approximate dynamic programming.}
To control state-space growth, we introduce
\texttt{\red{AppDP}}, which limits the number of states
retained in the sparse DP frontier.
It follows the same allocation transitions as the exact
DP but prunes excess states whenever the prescribed
state limit is exceeded.
This limits the number of states carried forward to
subsequent iterations, at the cost of potentially
discarding states needed for an optimal allocation.
Thus, \texttt{\red{AppDP}} (with a frontier cap of 100 states) preserves allocation feasibility but does not generally guarantee optimality.

\subsection{Mixed-Integer Linear Programming}
\label{sec:method_milp}

The DP explicitly represents multidimensional resource budgets.
To avoid constructing this resource-state space, we express
CMA as a compact MILP.
For each candidate $i$, introduce a binary selection variable
$y_i$, and define its standalone shot bound as
$M_i=M_i(\mathbf{B})$.
Every feasible allocation satisfies $x_i\leq M_i$, because
candidate $i$ alone cannot consume more Bell pairs than any
of its supporting links provides.

The equivalent MILP is
\begin{align}
\max_{x,y}\quad
& \sum_{i\in C}x_i
\label{eq:milp_obj}\\
\mathrm{s.t.}\quad
& \sum_{i\in C}a_{i,e}x_i\leq R_e,
&& \forall e\in E^+,
\label{eq:milp_resource}\\
& y_i\leq x_i\leq M_i y_i,
&& \forall i\in C,
\label{eq:milp_link}\\
& \sum_{i\in C}y_i\leq K,
\label{eq:milp_cardinality}\\
& x_i\in\mathbb{Z}_{\geq0},\quad y_i\in\{0,1\},
&& \forall i\in C.
\label{eq:milp_domain}
\end{align}

Constraint~\eqref{eq:milp_link} ensures that $y_i=1$ if and
only if candidate $i$ receives positive shots.
Consequently, Constraint~\eqref{eq:milp_cardinality} enforces
the selected-candidate limit.
Any feasible CMA allocation can be extended to a feasible
MILP solution by setting $y_i=1$ if $x_i>0$ and $y_i=0$ otherwise.
Conversely, every feasible MILP solution satisfies the original
resource and candidate-selection constraints.
The reformulation therefore preserves both feasibility
and the optimal objective value.

\paragraph{MILP complexity.}
The formulation contains $m$ integer shot variables, $m$ binary selection variables, and $d+2m+1$ linear
inequalities, excluding variable bounds.
Unlike the full DP resource grid, the number of variables and constraints does not grow with the numerical link capacities.
Larger capacities instead change the constraint coefficients or right-hand sides and their encoding lengths.

This compact representation does not imply polynomial-time solvability: CMA remains NP-hard, and exact MILP solution
may require extensive combinatorial search. Nevertheless, the formulation allows a solver to exploit relaxations, bounds, and pruning without explicitly enumerating all capacity states.

\subsection{\red{Lightweight Greedy Allocation}}
\label{sec:method_Efficient}

Although the preceding methods can improve resource utilization through joint allocation, their computational overhead can become substantial: DP may require a large multidimensional resource frontier, while MILP relies on potentially expensive integer optimization. Meanwhile, communication resources may remain underutilized when execution is constrained by bottleneck links, leaving residual capacity on other links. These considerations motivate \texttt{\red{Effi}}, a lightweight allocation method that incrementally exploits the remaining per-link resources without solving the full joint optimization problem. Its goal is to reduce allocation overhead while retaining competitive solution quality.

Let $R_e^{(t)}$ be the residual capacity of link $e$ before iteration $t$, with
$R_e^{(0)}=R_e$.  Under this residual vector, candidate $i$ can execute at most
\begin{equation}
s_i^{(t)}=
\min_{e:a_{i,e}>0}
\left\lfloor\frac{R_e^{(t)}}{a_{i,e}}\right\rfloor
\label{eq:efficient-increment}
\end{equation}
additional shots.  The algorithm selects a candidate maximizing
$s_i^{(t)}$, assigns all of its feasible shots, and subtracts their per-link
cost from the residual vector.

\begin{algorithm}[h]
\caption{\red{Effi}: progressive compiled loop}
\label{alg:efficient-allocation}
\begin{algorithmic}[1]
\Require Candidate pool, costs $a_{i,e}$, capacities $R_e$, limit $K$
\State $R_e^{\mathrm{res}}\gets R_e$ for every link $e$;
       $x_i\gets0$ for every candidate $i$
\For{$t=1,\ldots,K$}
    \ForAll{candidates $i$}
        \State $s_i\gets
        \displaystyle\min_{e:a_{i,e}>0}
        \left\lfloor R_e^{\mathrm{res}}/a_{i,e}\right\rfloor$
    \EndFor
    \State Choose candidate $i^\star$ with the largest $s_i$
    \If{$s_{i^\star}=0$}
        \State \textbf{break}
    \EndIf
    \State $x_{i^\star}\gets x_{i^\star}+s_{i^\star}$
    \ForAll{links $e$}
        \State $R_e^{\mathrm{res}}\gets
        R_e^{\mathrm{res}}-a_{i^\star,e}s_{i^\star}$
    \EndFor
\EndFor
\State \Return $x$
\end{algorithmic}
\end{algorithm}

Each update in Algorithm~\ref{alg:efficient-allocation} is feasible because Eq.~\eqref{eq:efficient-increment} guarantees
$a_{i^\star,e}s_{i^\star}\le R_e^{\mathrm{res}}$ on every link.  The first iteration is exactly the single-candidate baseline, and later iterations only add feasible shots, so \texttt{\red{Effi}} never completes fewer shots than that baseline.  It is nevertheless a heuristic: an early locally best candidate can consume capacity needed by a better later combination, so unlike \texttt{MILP} it provides no global-optimality certificate.  Its running time is $O(Km|E|)$ for $m$ compiled candidates, linear in the candidate pool and number of links at each of the at most $K$ iterations.

\begin{figure*}[t]
\centering
\includegraphics[width=0.8\textwidth]{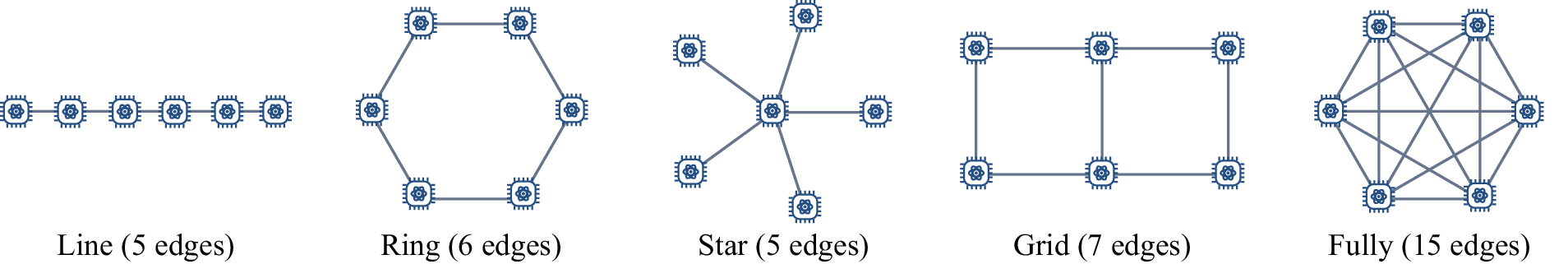}
\Description{Line, Ring, Star, two-by-three Grid, and Fully Connected
interconnection diagrams, each containing six numbered QPU nodes.}
\caption{Six-QPU network topologies used in the experiments.}
\label{fig:6qpu-topologies}
\end{figure*}


\begin{table*}[!t]
\centering
\caption{Shot-allocation optimization results for \textsc{vqe} across six-QPU network topologies.}
\label{tab:6qpu-opt-vqe}
\resizebox{\textwidth}{!}{%
\begin{tabular}{@{}l*{5}{rrrr}@{}}
\toprule
& \multicolumn{4}{c}{\textbf{Line}}
& \multicolumn{4}{c}{\textbf{Ring}}
& \multicolumn{4}{c}{\textbf{Star}}
& \multicolumn{4}{c}{\textbf{Grid}}
& \multicolumn{4}{c}{\textbf{Fully Connected}} \\
\cmidrule(lr){2-5}
\cmidrule(lr){6-9}
\cmidrule(lr){10-13}
\cmidrule(lr){14-17}
\cmidrule(lr){18-21}
Method
& $S$ & $\Delta$ & Utl.  & Time (s)
& $S$ & $\Delta$ & Utl.  & Time (s)
& $S$ & $\Delta$ & Utl.  & Time (s)
& $S$ & $\Delta$ & Utl.  & Time (s)
& $S$ & $\Delta$ & Utl.  & Time (s) \\
\midrule
\texttt{Baseline}
& 100 & -- & 60.0\% & 0.0004
& 100 & -- & 50.0\% & 0.0003
& 100 & -- & 20.0\% & 0.0003
& 100 & -- & 14.3\% & 0.0003
& 100 & -- & 6.7\% & 0.0003 \\
\texttt{\red{AppDP}}
& 199 & 99\% & 99.4\% & 0.3343
& 216 & 116\% & 99.7\% & 0.4358
& 230 & 130\% & 98.0\% & 0.2878
& 348 & 248\% & 99.4\% & 0.6328
& 450 & 350\% & 56.7\% & 1.5474 \\
\texttt{MILP}
& 199 & 99\% & 99.4\% & 0.4373
& 221 & 121\% & 99.3\% & 0.1508
& 233 & 133\% & 99.6\% & 0.2597
& 349 & 249\% & 99.6\% & 0.0795
& 450 & 350\% & 76.7\% & 0.3962 \\
\texttt{\red{Effi}}
& 150 & 50\% & 100.0\% & 0.0009
& 150 & 50\% & 83.3\% & 0.0007
& 225 & 125\% & 90.0\% & 0.0011
& 266 & 166\% & 99.7\% & 0.0013
& 450 & 350\% & 56.7\% & 0.0020 \\
\bottomrule
\end{tabular}%
}
\end{table*}

\section{Experiment}

We evaluate resource-aware candidate combination in two parts. The
\emph{optimization} part measures how much a joint allocation over multiple compiled candidates raises the executable shot rate over the single candidate, at what Bell-pair utilization and solver cost. The \emph{implementation} part then executes the selected allocations---in a distributed-circuit simulator and on IBM hardware---to measure the output fidelity they deliver under an equal Bell-pair budget. Throughout we compare four allocation methods: the single-candidate \texttt{Baseline}, the frontier-capped approximation dynamic programming \texttt{\red{AppDP}}, the integer program \texttt{MILP}, and the greedy
\texttt{\red{Effi}} heuristic, all receiving the same candidate set per instance under a default candidate limit $K=8$. We present the six-QPU simulation study
(Section~\ref{sec:eval-6qpu}), the IBM hardware validation (Section~\ref{sec:hardware}), and the scalability study with up to 36 QPUs (Section~\ref{sec:eval-scalability}).

\subsection{Experimental Setup}
\label{sec:eval-setup}

\paragraph{Resource model and parameters.}
Each inter-QPU link independently supplies Bell pairs at a rate of $R_e=100$ pairs/s. Links incident on the same QPU operate concurrently, and a multi-hop remote operation consumes one Bell pair on every traversed link. For each instance, candidates are generated once using the same partitioning, and all methods receive the identical candidate set. Candidate-generation time is excluded from the reported optimization runtime.

\paragraph{Topologies.}
Figure~\ref{fig:6qpu-topologies} shows the interconnection topologies we evaluate:
Line, Ring, Star, Grid, and Fully Connected (F.C.). Each edge is an independent Bell-pair generation resource, so the achievable throughput depends jointly on the connectivity and on how the selected candidates distribute their communication demand across links.

\paragraph{Circuit workloads.}
The simulation study uses six 10-qubit circuit workloads:
VQE~\cite{peruzzo2014vqe},
QAOA~\cite{farhi2014qaoa},
GHZ~\cite{greenberger1989going},
BITCODE and TFIM~\cite{Tomesh2022SupermarQAS},
and FULLADDER~\cite{li_qasmbench}.
The IBM hardware experiments use 20-qubit VQE, QAOA, and TFIM circuits, and the scalability study uses the same three workloads.

\paragraph{Optimization metrics.}
We report the optimized shot rate $S=\sum_i x_i$ (shots/s), its improvement ($\Delta$ ) over the single-candidate baseline, global Bell-pair utilization, the number of selected candidates $n$, and runtime:
\begin{equation*}
\Delta =\frac{S-S_{\mathrm{Base}}}{S_{\mathrm{Base}}},
\qquad
\mathrm{BP Utl.}
=
\frac{\sum_e\sum_i a_{i,e}x_i}{\sum_e R_e}.
\end{equation*}
Runtime measures the complete optimization procedure of each method, including model construction, optimization, and result parsing.

\paragraph{Implementation metric.}
The implementation part reports output fidelity. For a given allocation we merge the measurement counts of all funded candidates into an empirical distribution $q$ and compare it with the ideal distribution $p$ by the Hellinger fidelity coefficient
\begin{equation*}
F(p,q)=\textstyle\sum_z \sqrt{p(z)\,q(z)} .
\end{equation*}
Merging the counts lets several candidates jointly estimate the same logical output, so a larger executable shot budget reduces sampling error and raises $F$. We measure $F$ as the shot (and Bell-pair) budget of a fixed allocation is scaled to $T$ seconds of entanglement generation.

\subsection{Simulation Results}
\label{sec:eval-6qpu}

\subsubsection{Throughput Optimization Results across Interconnect Topologies}
\label{sec:eval-6qpu-throughput}

Table~\ref{tab:6qpu-opt-vqe} reports the optimization results across the five six-QPU interconnect topologies with VQE circuit workload. Across all topologies, multi-candidate allocation substantially improves the execution shot rate over the single-candidate baseline. The baseline is limited by the most heavily used link of a single compilation and therefore leaves capacity on other links unused. Combining multiple compilations with different per-link demands allows this residual Bell-pair capacity to support additional shots, resulting in higher throughput and generally higher network utilization.

Among the optimization methods, \texttt{MILP} achieves the highest or tied-highest execution shot rate across all topologies. \texttt{\red{AppDP}} obtains results close to \texttt{MILP}, but requires substantially longer optimization time in several cases. \texttt{\red{Effi}} also provides considerable throughput improvement over the baseline while requiring the shortest runtime, close to baseline runtime. These results show that multi-compilation allocation improves throughput across different interconnect structures, while \texttt{\red{Effi}} provides a low-overhead alternative when optimization time is a primary concern.

\begin{table}[t]
\centering
\caption{Simulation fidelity results for 10-qubit circuits on six-QPU systems, averaged across QPU topologies.}
\label{tab:6qpu-fid}
\scriptsize
\setlength{\tabcolsep}{3pt}
\resizebox{\columnwidth}{!}{%
\begin{tabular}{@{}ll rr rrrrr@{}}
\toprule
& & & & \multicolumn{5}{c}{Merged fidelity $F$ at budget $T$ (s)} \\
\cmidrule(lr){5-9}
  & Method & $S$ & $n$ & $T{=}1$ & $T{=}2$ & $T{=}4$ & $T{=}6$ & $T{=}8$ \\
\midrule
\multirow{4}{*}{\textsc{vqe}} & \texttt{Baseline} & 100 & 1.0 & 0.400 & 0.522 & 0.659 & 0.735 & 0.780 \\
 & \texttt{\red{AppDP}} & 288.6 & 5.0 & 0.524 & 0.657 & 0.784 & 0.843 & 0.878 \\
 & \texttt{MILP} & 290.4 & 5.2 & 0.527 & 0.667 & 0.788 & 0.847 & 0.879 \\
 & \texttt{\red{Effi}} & 248.2 & 2.8 & 0.486 & 0.618 & 0.752 & 0.814 & 0.854 \\
\midrule
\multirow{4}{*}{\textsc{qaoa}} & \texttt{Baseline} & 50 & 1.0 & 0.253 & 0.351 & 0.472 & 0.557 & 0.614 \\
 & \texttt{\red{AppDP}} & 168.8 & 4.2 & 0.438 & 0.577 & 0.717 & 0.794 & 0.838 \\
 & \texttt{MILP} & 168.8 & 5.0 & 0.438 & 0.575 & 0.717 & 0.796 & 0.840 \\
 & \texttt{\red{Effi}} & 168.8 & 4.2 & 0.438 & 0.574 & 0.718 & 0.794 & 0.840 \\
\midrule
\multirow{4}{*}{\textsc{ghz}} & \texttt{Baseline} & 10 & 1.0 & 0.840 & 0.845 & 0.828 & 0.827 & 0.833 \\
 & \texttt{\red{AppDP}} & 57.5 & 5.8 & 0.835 & 0.835 & 0.838 & 0.832 & 0.835 \\
 & \texttt{MILP} & 57.5 & 5.8 & 0.834 & 0.840 & 0.840 & 0.834 & 0.835 \\
 & \texttt{\red{Effi}} & 57.5 & 5.8 & 0.834 & 0.838 & 0.838 & 0.840 & 0.835 \\
\midrule
\multirow{4}{*}{\shortstack{\textsc{bit}\\\textsc{code}}} & \texttt{Baseline} & 100 & 1.0 & 0.762 & 0.809 & 0.820 & 0.825 & 0.829 \\
 & \texttt{\red{AppDP}} & 575 & 5.8 & 0.831 & 0.835 & 0.837 & 0.837 & 0.838 \\
 & \texttt{MILP} & 575 & 5.8 & 0.833 & 0.835 & 0.837 & 0.838 & 0.837 \\
 & \texttt{\red{Effi}} & 575 & 5.8 & 0.832 & 0.837 & 0.835 & 0.838 & 0.838 \\
\midrule
\multirow{4}{*}{\textsc{tfim}} & \texttt{Baseline} & 25 & 1.0 & 0.321 & 0.416 & 0.523 & 0.585 & 0.631 \\
 & \texttt{\red{AppDP}} & 143.8 & 5.8 & 0.581 & 0.689 & 0.783 & 0.824 & 0.851 \\
 & \texttt{MILP} & 143.8 & 5.8 & 0.582 & 0.693 & 0.780 & 0.822 & 0.849 \\
 & \texttt{\red{Effi}} & 143.8 & 5.8 & 0.580 & 0.690 & 0.783 & 0.826 & 0.851 \\
\midrule
\multirow{4}{*}{\shortstack{\textsc{full}\\\textsc{adder}}} & \texttt{Baseline} & 10 & 1.0 & 0.105 & 0.134 & 0.190 & 0.232 & 0.270 \\
 & \texttt{\red{AppDP}} & 42.5 & 5.0 & 0.205 & 0.277 & 0.377 & 0.437 & 0.490 \\
 & \texttt{MILP} & 42.5 & 5.0 & 0.197 & 0.277 & 0.374 & 0.440 & 0.492 \\
 & \texttt{\red{Effi}} & 37.8 & 4.5 & 0.184 & 0.259 & 0.358 & 0.414 & 0.457 \\
\bottomrule
\end{tabular}%
}
\end{table}

\subsubsection{End-to-End Simulation Fidelity}
\label{sec:eval-6qpu-fidelity}

\begin{figure*}[t]
\centering
\includegraphics[width=\textwidth]{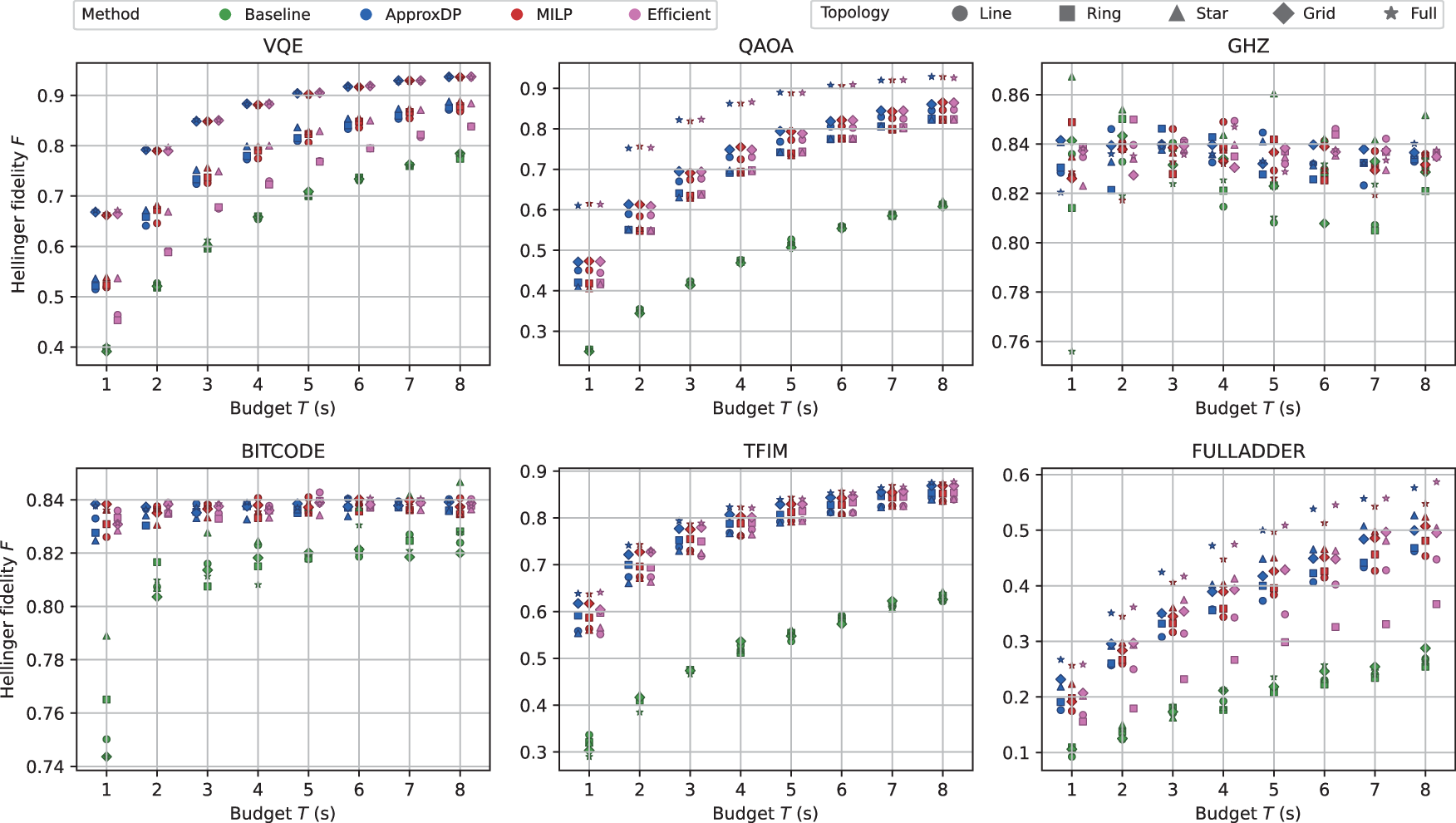}
\caption{Simulation fidelity versus time budget for six 10-qubit circuit workloads. Colors denote allocation methods, markers denote six-QPU topologies.}
\label{fig:fidelity}
\end{figure*}

This experiment evaluates the output fidelity of the optimized allocations on the same six-QPU topologies used in Section~\ref{sec:eval-6qpu-throughput}. Each allocation is fixed while the time budget is scaled from $T=1$ to $8$ seconds, and each reported fidelity is averaged over 10 independent simulation replicates. For GHZ, we set $R_e=10$ pairs/s because its output distribution requires substantially fewer shots to converge.

\paragraph{Average fidelity.}
Table~\ref{tab:6qpu-fid} reports the fidelity of each optimized allocation from
Table~\ref{tab:6qpu-opt-vqe} as its shot
budget is scaled from $T{=}1$ to $8$ seconds; the two share the same workloads, methods, and topologies, so each fidelity row is the measured outcome of the
corresponding allocation.
At the same time budget, multi-candidate allocation executes more shots than the single-candidate baseline and therefore improves fidelity for sampling-limited workloads. The effect is clear for VQE, QAOA, TFIM, and FULLADDER. For example, QAOA improves from $0.25$ to about $0.44$ at $T=1$ and from $0.61$ to about $0.84$ at $T=8$, while TFIM improves from $0.32$ to about $0.58$ at $T=1$ and from $0.63$ to about $0.85$ at $T=8$. VQE shows the same trend, with the multi-candidate methods reaching about $0.88$ at $T=8$ compared with $0.78$ for the baseline. FULLADDER also benefits consistently, although its overall fidelity remains lower.

When the optimized methods achieve similar shot rates, their fidelity is also nearly identical. This is particularly visible for QAOA, GHZ, BITCODE, and TFIM, where \texttt{\red{AppDP}}, \texttt{MILP}, and \texttt{\red{Effi}} produce almost overlapping fidelity values. In contrast, GHZ remains nearly flat around $0.83$ across all budgets, while BITCODE quickly approaches a plateau near $0.84$. For these workloads, the fidelity benefit from additional execution shots has largely saturated.


\begin{table*}[t]
\centering
\caption{Optimization results for VQE-20, QAOA-20, and TFIM-20 on the
nine-QPU partition of \texttt{ibm\_pittsburgh}}
\label{tab:exp2-s60-alloc}
\resizebox{0.9\textwidth}{!}{%
\begin{tabular}{@{}l*{3}{rrrrr}@{}}
\toprule
& \multicolumn{5}{c}{\textbf{VQE-20}}
& \multicolumn{5}{c}{\textbf{QAOA-20}}
& \multicolumn{5}{c}{\textbf{TFIM-20}} \\
\cmidrule(lr){2-6}
\cmidrule(lr){7-11}
\cmidrule(lr){12-16}
Method
& $S$ & $\Delta$ & $n$ & Utl. & Time (s)
& $S$ & $\Delta$ & $n$ & Utl. & Time (s)
& $S$ & $\Delta$ & $n$ & Utl. & Time (s) \\
\midrule
\texttt{Baseline}
& 100 & -- & 1 & 16.7\% & 0.0002
& 50 & -- & 1 & 33.3\% & 0.0002
& 25 & -- & 1 & 16.7\% & 0.0002 \\
\texttt{\red{AppDP}}
& 450 & 350\% & 5 & 87.5\% & 1.5730
& 132 & 164\% & 8 & 91.0\% & 1.5616
& 150 & 500\% & 6 & 100.0\% & 0.8713 \\
\texttt{MILP}
& 450 & 350\% & 6 & 93.8\% & 0.6420
& 168 & 236\% & 8 & 97.5\% & 0.5823
& 150 & 500\% & 8 & 100.0\% & 0.3116 \\
\texttt{\red{Effi}}
& 450 & 350\% & 5 & 87.5\% & 0.0011
& 125 & 150\% & 3 & 83.3\% & 0.0008
& 125 & 400\% & 5 & 83.3\% & 0.0012 \\
\bottomrule
\end{tabular}%
}
\end{table*}

\begin{figure*}[t]
    \centering
    \includegraphics[width=\textwidth]{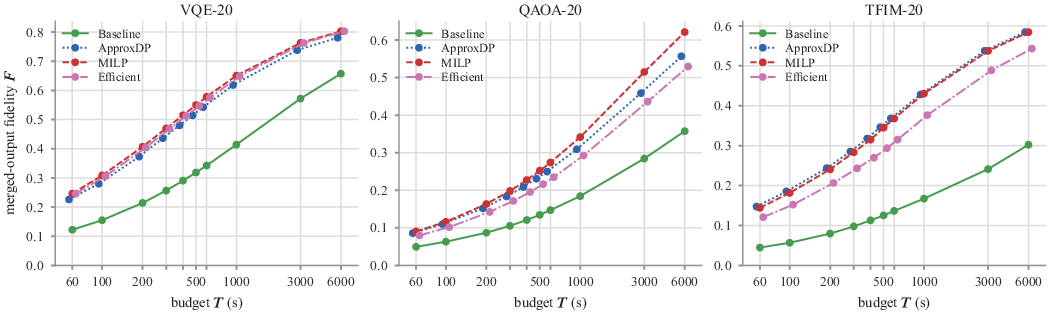}
\caption{Quantum measured fidelity on hardware for VQE-20,
QAOA-20, and TFIM-20 as the time budget $T$ increases.}
    \label{fig:exp2-s60-fidelity}
\end{figure*}

\paragraph{Topology-resolved fidelity.}
Figure~\ref{fig:fidelity} further resolves the fidelity results by inter-QPU topology. For VQE, QAOA, TFIM, and FULLADDER, the multi-candidate methods remain consistently above the single-candidate baseline as the budget increases, showing that their higher shot throughput translates into higher output fidelity. The gap is most pronounced at smaller budgets and gradually narrows as the fidelity benefit from additional execution shots begins to saturate.

\texttt{\red{AppDP}} and \texttt{MILP} generally produce closely clustered fidelity values, while \texttt{\red{Effi}} is also comparable on most topology--budget combinations but can trail when its optimized shot rate is lower.

\subsection{Hardware Results}
\label{sec:hardware}

To validate the framework on real quantum hardware, we repeat the
optimization-and-execution workflow on the 156-qubit
\texttt{ibm\_pittsburgh}. Since the device is monolithic, we emulate the DQC
setting by partitioning it into nine logical QPUs arranged as a $3\times3$ grid
with twelve inter-QPU links. Bell pairs for remote operations are prepared
on-chip using communication qubits assigned to neighboring logical QPUs. We
evaluate VQE, QAOA, and TFIM circuits under the same multi-QPU execution model
described in Section~\ref{sec:eval-setup}.

\subsubsection{Throughput Optimization Results on Hardware}
Table~\ref{tab:exp2-s60-alloc} reports the optimization results on the hardware partition. Multi-candidate allocation substantially increases the execution shot rate and Bell-pair utilization over the single-candidate baseline for all three
workloads, although the achievable improvement differs across circuits. Under \texttt{MILP}, throughput improves by 350\%, 236\%, and 500\% over the baseline for VQE-20, QAOA-20, and TFIM-20 respectively, indicating that different workloads provide different optimization opportunities.

Across optimization methods, \texttt{MILP} achieves the highest or tied-highest execution shot rate in every case. \texttt{\red{AppDP}} remains close to \texttt{MILP} but generally requires longer optimization time, while \texttt{\red{Effi}} achieves substantial throughput improvement with a solver runtime comparable to the \texttt{Baseline}.

\subsubsection{Quantum Hardware Measured Fidelity}
Figure~\ref{fig:exp2-s60-fidelity} reports the measured output fidelity as the time budget increases. At the same time budget, all multi-candidate methods consistently achieve higher fidelity than the single-candidate baseline, showing that the additional executable shots translate into higher measured
fidelity on hardware. At $T=6000$\,s, \texttt{MILP} improves VQE-20 from $0.657$
to $0.803$, QAOA-20 from $0.357$ to $0.621$, and TFIM-20 from $0.302$ to $0.584$,
corresponding to relative fidelity improvements of $22.1\%$, $73.8\%$, and
$93.3\%$, respectively.

The differences among the optimized methods also follow their allocated shot
rates. For VQE-20, all three multi-candidate methods use the same shot rate of
450 shots/s and converge to similar fidelity, reaching $0.780$--$0.803$ at
$T=6000$\,s. For QAOA-20, \texttt{MILP} has the highest shot rate
(168 shots/s) and reaches $0.621$, compared with $0.557$ for \texttt{\red{AppDP}}
and $0.529$ for \texttt{\red{Effi}}. For TFIM-20, \texttt{\red{AppDP}} and
\texttt{MILP} both execute 150 shots/s and nearly overlap at $0.584$, whereas
\texttt{\red{Effi}}, at 125 shots/s, reaches $0.543$.

\begin{table*}[!t]
\centering
\footnotesize
\setlength{\tabcolsep}{2.5pt}
\renewcommand{\arraystretch}{1.05}
\caption{VQE shot allocation on 12--36 QPUs across four QPU topologies.}
\label{tab:exp2-scalability-vqe}
\resizebox{\textwidth}{!}{%
\begin{tabular}{@{}llrrrlrrrlrrrlrrrl@{}}
\toprule
 &  & \multicolumn{4}{c}{\textbf{12 QPUs, 50-qubit}} & \multicolumn{4}{c}{\textbf{20 QPUs, 80-qubit}} & \multicolumn{4}{c}{\textbf{30 QPUs, 120-qubit}} & \multicolumn{4}{c}{\textbf{36 QPUs, 144-qubit}} \\
\cmidrule(lr){3-6}\cmidrule(lr){7-10}\cmidrule(lr){11-14}\cmidrule(lr){15-18}
 & Method & $S$ & $\Delta$ (\%) & Utl. & Time (s) & $S$ & $\Delta$ (\%) & Utl. & Time (s) & $S$ & $\Delta$ (\%) & Utl. & Time (s) & $S$ & $\Delta$ (\%) & Utl. & Time (s) \\
\midrule
Line & \texttt{Baseline} & 33 & 0.0\% & 39.0\% & 0.001 & 16 & 0.0\% & 67.4\% & 0.001 & 12 & 0.0\% & 62.9\% & 0.002 & 8 & 0.0\% & 40.9\% & 0.002 \\
 & \texttt{\red{AppDP}} & 53 & 60.6\% & 74.7\% & 0.626 & 20 & 25.0\% & 68.6\% & 1.615 & 14 & 16.7\% & 79.7\% & 1.610 & 10 & 25.0\% & 63.6\% & 3.650 \\
 & \texttt{MILP} & 55 & 66.7\% & 78.2\% & 0.146 & 20 & 25.0\% & 68.6\% & 0.193 & 15 & 25.0\% & 74.0\% & 0.149 & 10 & 25.0\% & 73.2\% & 0.087 \\
 & \texttt{\red{Effi}} & 49 & 48.5\% & 68.1\% & 0.002 & 16 & 0.0\% & 67.4\% & 0.002 & 12 & 0.0\% & 62.9\% & 0.003 & 9 & 12.5\% & 51.7\% & 0.006 \\
\midrule
Ring & \texttt{Baseline} & 50 & 0.0\% & 54.2\% & 0.001 & 20 & 0.0\% & 46.0\% & 0.001 & 16 & 0.0\% & 48.0\% & 0.002 & 10 & 0.0\% & 50.0\% & 0.002 \\
 & \texttt{\red{AppDP}} & 66 & 32.0\% & 68.7\% & 1.299 & 27 & 35.0\% & 81.8\% & 1.654 & 21 & 31.2\% & 77.6\% & 3.179 & 15 & 50.0\% & 77.7\% & 5.227 \\
 & \texttt{MILP} & 78 & 56.0\% & 89.8\% & 0.099 & 29 & 45.0\% & 87.2\% & 0.419 & 22 & 37.5\% & 82.9\% & 0.102 & 16 & 60.0\% & 89.1\% & 0.083 \\
 & \texttt{\red{Effi}} & 67 & 34.0\% & 71.5\% & 0.002 & 20 & 0.0\% & 46.0\% & 0.002 & 17 & 6.2\% & 51.2\% & 0.004 & 13 & 30.0\% & 65.9\% & 0.006 \\
\midrule
Star & \texttt{Baseline} & 50 & 0.0\% & 31.8\% & 0.000 & 25 & 0.0\% & 28.9\% & 0.001 & 25 & 0.0\% & 25.0\% & 0.001 & 25 & 0.0\% & 29.3\% & 0.002 \\
 & \texttt{\red{AppDP}} & 112 & 124.0\% & 84.7\% & 2.244 & 67 & 168.0\% & 83.9\% & 8.307 & 59 & 136.0\% & 70.4\% & 11.074 & 47 & 88.0\% & 73.9\% & 17.999 \\
 & \texttt{MILP} & 118 & 136.0\% & 89.9\% & 0.118 & 70 & 180.0\% & 85.2\% & 2.804 & 66 & 164.0\% & 79.8\% & 3.568 & 54 & 116.0\% & 77.4\% & 179.474 \\
 & \texttt{\red{Effi}} & 84 & 68.0\% & 57.5\% & 0.002 & 50 & 100.0\% & 59.6\% & 0.004 & 51 & 104.0\% & 67.1\% & 0.005 & 33 & 32.0\% & 45.1\% & 0.006 \\
\midrule
Grid & \texttt{Baseline} & 50 & 0.0\% & 32.4\% & 0.001 & 33 & 0.0\% & 26.6\% & 0.001 & 33 & 0.0\% & 37.7\% & 0.002 & 20 & 0.0\% & 27.7\% & 0.003 \\
 & \texttt{\red{AppDP}} & 101 & 102.0\% & 67.4\% & 3.408 & 56 & 69.7\% & 60.5\% & 8.455 & 47 & 42.4\% & 61.3\% & 10.908 & 38 & 90.0\% & 61.3\% & 17.682 \\
 & \texttt{MILP} & 120 & 140.0\% & 86.6\% & 0.363 & 78 & 136.4\% & 75.8\% & 1.295 & 51 & 54.5\% & 73.9\% & 0.824 & 44 & 120.0\% & 70.5\% & 4.320 \\
 & \texttt{\red{Effi}} & 75 & 50.0\% & 45.6\% & 0.002 & 62 & 87.9\% & 58.6\% & 0.005 & 34 & 3.0\% & 39.1\% & 0.006 & 28 & 40.0\% & 44.1\% & 0.009 \\
\bottomrule
\end{tabular}%
}
\end{table*}

The throughput gain can also be viewed as a reduction in the time required to
reach a target fidelity. Using the baseline fidelity at $T=6000$\,s as the target,
all multi-candidate methods reach or exceed it by $T=3000$\,s for VQE-20 and
QAOA-20, and by only $T=600$\,s for TFIM-20. Thus, at the sampled budgets,
multi-candidate allocation reduces the time to reach the baseline target fidelity
by at least $2\times$ for VQE-20 and QAOA-20 and by $10\times$ for TFIM-20.

\subsection{Scalability Results}
\label{sec:eval-scalability}

To evaluate scalability, we increase the distributed architecture from 12 to 36 QPUs, using 50-, 80-, 120-, and 144-qubit circuits on 12, 20, 30, and 36 QPUs, respectively. We evaluate VQE, QAOA, and TFIM across four topologies and compare the four allocation methods in terms of shot throughput, Bell-pair utilization, and optimization runtime.

\subsubsection{Results for VQE}

\begin{table*}[t]
\centering

\begin{subtable}{\textwidth}
\centering
\footnotesize
\setlength{\tabcolsep}{2.5pt}
\renewcommand{\arraystretch}{1.05}
\resizebox{\textwidth}{!}{%
\begin{tabular}{@{}lrrrlrrrlrrrlrrrl@{}}
\toprule
 & \multicolumn{4}{c}{\textbf{12 QPUs, 50-qubit}} & \multicolumn{4}{c}{\textbf{20 QPUs, 80-qubit}} & \multicolumn{4}{c}{\textbf{30 QPUs, 120-qubit}} & \multicolumn{4}{c}{\textbf{36 QPUs, 144-qubit}} \\
\cmidrule(lr){2-5}\cmidrule(lr){6-9}\cmidrule(lr){10-13}\cmidrule(lr){14-17}
Method & $S$ & $\Delta$ (\%) & Utl. & Time (s) & $S$ & $\Delta$ (\%) & Utl. & Time (s) & $S$ & $\Delta$ (\%) & Utl. & Time (s) & $S$ & $\Delta$ (\%) & Utl. & Time (s) \\
\midrule
\texttt{Baseline} & 62.5 & 0.0\% & 50.3\% & 0.001 & 52.0 & 0.0\% & 43.7\% & 0.001 & 36.3 & 0.0\% & 42.6\% & 0.002 & 35.3 & 0.0\% & 59.3\% & 0.002 \\
\texttt{\red{AppDP}} & 87.8 & 50.0\% & 74.8\% & 2.518 & 74.3 & 50.9\% & 72.3\% & 3.931 & 53.3 & 38.5\% & 74.5\% & 6.264 & 45.5 & 23.7\% & 67.9\% & 8.318 \\
\texttt{MILP} & 104.3 & 74.3\% & 90.9\% & 0.472 & 88.8 & 74.4\% & 85.1\% & 0.680 & 61.8 & 58.3\% & 78.0\% & 24.243 & 57.0 & 49.3\% & 82.0\% & 153.059 \\
\texttt{\red{Effi}} & 82.5 & 40.0\% & 73.5\% & 0.002 & 66.8 & 33.8\% & 57.2\% & 0.003 & 36.3 & 0.0\% & 42.6\% & 0.003 & 35.3 & 0.0\% & 59.3\% & 0.005 \\
\bottomrule
\end{tabular}%
}
\caption{QAOA}
\label{tab:exp2-scalability-qaoa-avg}
\end{subtable}

\vspace{6pt}

\begin{subtable}{\textwidth}
\centering
\footnotesize
\setlength{\tabcolsep}{2.5pt}
\renewcommand{\arraystretch}{1.05}
\resizebox{0.95\textwidth}{!}{%
\begin{tabular}{@{}lrrrlrrrlrrrlrrrl@{}}
\toprule
 & \multicolumn{4}{c}{\textbf{12 QPUs, 50-qubit}} & \multicolumn{4}{c}{\textbf{20 QPUs, 80-qubit}} & \multicolumn{4}{c}{\textbf{30 QPUs, 120-qubit}} & \multicolumn{4}{c}{\textbf{36 QPUs, 144-qubit}} \\
\cmidrule(lr){2-5}\cmidrule(lr){6-9}\cmidrule(lr){10-13}\cmidrule(lr){14-17}
Method & $S$ & $\Delta$ (\%) & Utl. & Time (s) & $S$ & $\Delta$ (\%) & Utl. & Time (s) & $S$ & $\Delta$ (\%) & Utl. & Time (s) & $S$ & $\Delta$ (\%) & Utl. & Time (s) \\
\midrule
\texttt{Baseline} & 37.5 & 0.0\% & 57.0\% & 0.000 & 23.5 & 0.0\% & 48.4\% & 0.001 & 18.0 & 0.0\% & 37.8\% & 0.002 & 16.3 & 0.0\% & 53.8\% & 0.002 \\
\texttt{\red{AppDP}} & 56.0 & 64.0\% & 84.0\% & 2.356 & 39.3 & 61.0\% & 72.0\% & 4.094 & 28.0 & 49.5\% & 64.1\% & 6.575 & 23.5 & 38.8\% & 64.7\% & 7.647 \\
\texttt{MILP} & 63.3 & 85.0\% & 92.6\% & 0.423 & 48.3 & 93.8\% & 85.3\% & 1.036 & 33.0 & 70.6\% & 76.9\% & 14.012 & 30.5 & 69.2\% & 78.1\% & 33.684 \\
\texttt{\red{Effi}} & 48.8 & 45.0\% & 73.6\% & 0.001 & 37.3 & 48.5\% & 66.6\% & 0.004 & 18.5 & 4.2\% & 39.6\% & 0.004 & 18.3 & 10.0\% & 57.1\% & 0.006 \\
\bottomrule
\end{tabular}%
}
\caption{TFIM}
\label{tab:exp2-scalability-tfim-avg}
\end{subtable}

\caption{Shot allocation on 12--36 QPUs, averaged across four QPU topologies (Line/Ring/Star/Grid).}
\label{tab:scal-qaoa-tfim}
\end{table*}

\paragraph{Throughput and topology effects.}
Table~\ref{tab:exp2-scalability-vqe} shows that combining multiple compiled candidates consistently improves shot throughput over the single-candidate \texttt{Baseline}, although the gain varies across methods, scales, and topologies. \texttt{MILP} achieves the highest or tied-highest throughput in all 16 configurations, with improvements of 25.0\%--180.0\% over \texttt{Baseline}. \texttt{\red{AppDP}} generally remains close to \texttt{MILP}, achieving on average 91.1\% of its throughput, while \texttt{\red{Effi}} achieves 75.4\%. The advantage of candidate combination is particularly pronounced on Star and Grid. At 12 QPUs, for example, \texttt{MILP} increases throughput from 50 to 118 shots/s on Star and from 50 to 120 shots/s on Grid (136.0\% and 140.0\%), whereas Line improves only from 33 to 55 shots/s (66.7\%); at 36 QPUs the gaps are similar, with 116.0\% and 120.0\% on Star and Grid against 25.0\% and 60.0\% on Line and Ring. The difference follows from how much capacity a single candidate leaves idle: \texttt{Baseline} utilizes only 25--38\% of the available Bell pairs on Star and Grid but 39--67\% on Line and Ring, so Star and Grid offer far more residual capacity for complementary candidates to absorb.

As the system grows from 12 QPUs with a 50-qubit circuit to 36 QPUs with a 144-qubit circuit, absolute throughput generally decreases because larger circuits consume more Bell pairs per shot. Nevertheless, the relative benefit of candidate combination persists at the largest scale: at 36 QPUs, \texttt{MILP} still more than doubles \texttt{Baseline} throughput on Star (25 to 54 shots/s) and Grid (20 to 44 shots/s), \texttt{\red{AppDP}} reaches 47 and 38 shots/s, and even \texttt{\red{Effi}}, whose gains are smaller and vanish in a few Line and Ring configurations, improves by 32.0\% and 40.0\% there.

\paragraph{Optimization runtime.}
The three methods trade throughput against runtime differently. \texttt{MILP} delivers the highest throughput, but its runtime is instance-dependent and can become large for difficult configurations, reaching 179.5\,s on the 36-QPU Star. \texttt{\red{AppDP}} operates at the second scale (at most 18.0\,s) while staying close to \texttt{MILP} in throughput. \texttt{\red{Effi}} remains at the millisecond scale in every configuration, close to \texttt{Baseline}, while still adding throughput through multi-candidate allocation.

\subsubsection{Results for QAOA and TFIM}

Table~\ref{tab:scal-qaoa-tfim} shows the same trends for QAOA and TFIM, averaged over the four topologies. \texttt{MILP} consistently achieves the highest throughput, while \texttt{\red{AppDP}} retains on average 86.6\% and 87.9\% of it for QAOA and TFIM, respectively. Although absolute throughput again decreases with scale, multi-candidate allocation remains beneficial at 36 QPUs, where \texttt{MILP} still improves over \texttt{Baseline} by 49.3\% for QAOA and 69.2\% for TFIM, and \texttt{\red{AppDP}} by 23.7\% and 38.8\%.

The runtime picture also carries over. \texttt{\red{Effi}}'s gains shrink with scale and vanish for QAOA at 30 and 36 QPUs, but its runtime stays in milliseconds throughout, whereas \texttt{\red{AppDP}} requires several seconds and \texttt{MILP} up to 153.1\,s for QAOA at 36 QPUs.



\section{Conclusion}
In DQC, entanglement is the scarce resource, and a single compilation leaves much of it idle. This work introduces multi-compilation execution, which runs logically equivalent compilations of a circuit concurrently and pools their samples, converting idle per-link Bell-pair capacity into additional shots. We formulate the underlying candidate-selection and shot-allocation problem as CMA, prove it NP-hard, and solve it with a dynamic program and its approximation (\texttt{\red{AppDP}}), a compact MILP, and a millisecond-scale greedy heuristic (\texttt{\red{Effi}}). Across six 10-qubit workloads, five topologies, IBM quantum hardware, and systems of up to 36 QPUs and 144-qubit circuits, multi-compilation raises shot throughput substantially over the best single compilation and translates that throughput into higher output fidelity under equal Bell-pair budgets. These results establish multi-compilation as a practical lever for managing entanglement in distributed quantum systems.

\clearpage
\bibliographystyle{ACM-Reference-Format}
\bibliography{ref}

\end{document}